\documentclass[12pt]{iopart}

\usepackage{iopams}
\usepackage{amsthm}
\usepackage{hyperref}

\newcommand{\C}[2]{{#1 \choose #2}}
\newcommand{\sgn}{\mathrm{sgn}}
\newcommand{\Ai}{\mathrm{Ai}}
\renewcommand{\O}[1]{\mathcal{O}\left(#1\right)}

\newtheorem{lemma}{Lemma}

\newcommand{\affiliation}[1]{\address{#1}}
\newcommand{\keywords}[1]{\vspace{10mm}\noindent\textbf{Keywords:} #1}
\newcommand{\tfrac}[2]{\mbox{\small$\frac{#1}{#2}$}}
\newcommand{\openone}{\mbox{{\small 1}$\!\!$1}}
\newcommand{\text}[1]{\mathrm{#1}}
\newcommand{\substack}[1]{\tiny\begin{array}{c}#1\end{array}}

\begin{document}

\title[Two-point function  of the free energy for a directed polymer]{Two-point generating function of the free energy for a directed polymer in a random medium}
\author{Sylvain Prolhac\footnote{prolhac@ma.tum.de}{} and Herbert Spohn\footnote{spohn@ma.tum.de}{}}
\affiliation{Zentrum Mathematik and Physik Department,\\Technische Universit\"at M\"unchen,\\D-85747 Garching, Germany}
\date{\today}

\begin{abstract}
We consider a 1+1 dimensional directed continuum polymer in a Gaussian delta-correlated space-time random potential. For this model the moments (= replica) of the partition function, $Z(x,t)$, can be expressed in terms of the attractive $\delta$-Bose gas on the line. Based on a recent study of the structure of the eigenfunctions, we compute the generating function for $Z(x_1,t)$, $Z(x_2,t)$ under a particular decoupling assumption and thereby extend recent results on the one-point generating function  of the free energy to two points. It is established that in the long time limit the fluctuations of the free energy are governed by the two-point distribution of the Airy process, which further supports that the long time behavior of the KPZ equation is the same as derived previously for lattice growth models.

\end{abstract}


\maketitle

\section{Introduction}
\label{Section introduction}
\setcounter{equation}{0}
Directed polymer in a random medium is a widely studied model in the statistical mechanics of disordered systems \cite{HHZ95.1,K07.1}. The polymer chain is immersed in a static random potential. In the directed version there is a singled out direction, also referred to as ``time", such that the polymer is constrained to move forward along the time direction. Prominent realizations are domain walls in two dimensional disordered magnets \cite{LFCMGLD98.1}, tear lines for paper sheets \cite{K92.1} and vortex lines in disordered superconductors \cite{BFGLV94.1}.
For spatial dimension $d\leq2$, the polymer is superdiffusive at any coupling strength, while for $d>2$, there is a weak coupling phase with diffusive behavior and a still little explored strong coupling phase, see \cite{CC10,L09} and references therein to earlier work.

Recently there has been considerable progress \cite{SS10.1,SS10.2,ACQ10,SS10.3,CLDR10.1,D10.1,D10.2} for a particular continuum version of the directed polymer in $1+1$ dimension with both endpoints fixed (point-to-point directed polymer). The free polymer is modeled by a continuum Brownian motion with bending coefficient $\gamma$. Then, in the presence of a disorder potential $V$, the point-to-point partition function of the polymer at inverse temperature $\beta$ reads
\begin{equation}
\label{Z(x,t)}
\fl\qquad\qquad
Z(x,t) = \int_{x(0)=0}^{x(t)=x}\mathcal{D}[x(\tau)]\,\exp\Big({-\beta\int_{0}^{t}\rmd\tau\left[\tfrac{1}{2}\gamma(\partial_{\tau}x(\tau))^{2}+V(x(\tau),\tau)\right]}\Big)\;.
\end{equation}
The partition function is the sum over all possible paths $x(\tau)$ of the polymer, starting at position $0$ at time $0$ and ending at position $x$ at time $t$. The energy of the polymer is the sum of the elastic bending energy, proportional to $\gamma$, and the potential energy obtained from summing the external potential $V$ along the polymer chain. The partition function $Z(x,t)$ is random as inherited from the randomness of the potential $V$, which is assumed to have a Gaussian distribution with mean $0$ and covariance
\begin{equation}
\label{<VV>}
\left\langle V(x,\tau)V(x',\tau')\right\rangle=D\,\delta(x-x')\delta(\tau-\tau')\;.
\end{equation}
The particular choice of the covariance (\ref{<VV>}) allows to express the $n$-th moment of $Z$ as a propagator matrix element of an $n$-particle attractive $\delta$-Bose gas on the line, a model which can be solved exactly by the Bethe ansatz. As will be explained in more detail below, the progress alluded to refers to an exact computation of the generating function for the partition sum $Z(x,t)$. In our contribution, this result will be extended to the generating function jointly of $Z(x_{1},t)$ and $Z(x_{2},t)$, \textit{i.e.} for two distinct positions $x_{1}$ and $x_{2}$ of the endpoint of the polymer at the same time $t$.

An additional interest in the continuum directed polymer comes from the connection to the Kardar-Parisi-Zhang (KPZ) equation \cite{KPZ86.1}, which is a stochastic evolution for a growing surface. If we denote the height profile by $h(x,t)$, then, in the conventional units, the one-dimensional KPZ equation reads
\begin{equation}
\label{KPZ}
\partial_{t}h(x,t)=\tfrac{1}{2}\lambda(\partial_{x}h(x,t))^{2}+\nu\,\partial_{x}^{2}h(x,t)+\eta(x,t)\;.
\end{equation}
Here, $\lambda$ is the strength of the nonlinear growth velocity, $\nu$ is the parameter governing the surface relaxation, and $\eta$ is a white noise with strength $\sqrt{D}$ modeling the random nucleation and deposition events at the surface.

The partition function of the directed polymer $Z(x,t)$ given in Eq. (\ref{Z(x,t)}) satisfies
\begin{equation}
\partial_{t}Z(x,t)=\frac{1}{2\beta\gamma}\,\partial_{x}^{2}Z(x,t)-\beta\,V(x,t)Z(x,t)\;,
\end{equation}
from which it follows that the free energy, defined by 
\begin{equation}
\label{F[Z]}
F(x,t)=-\frac{1}{\beta}\log Z(x,t)\;,
\end{equation}
is a solution of the KPZ equation (\ref{KPZ}) under the identification
\begin{equation}
h=-F\;,\quad\lambda=\gamma^{-1}\;,\quad\nu=(2\beta\gamma)^{-1}\;,\quad\eta=-V\;.
\end{equation}
In the context of surface growth, the joint distribution of the free energy is a natural quantity of interest: for given time $t$, it is the joint height statistics at the two spatial reference points $x_{1}$ and $x_{2}$, which in particular determines the height-height correlations at time $t$.

Our considerations are somewhat formal, since taken literally $\langle Z(x,t)\rangle = \infty$. In dimension $d=1$, this divergence can be easily taken care of by a suitable free energy renormalization, as discussed in detail in \cite{SS10.4}. After renormalization one finds that
\begin{equation}
\label{<Z>}
\langle Z(x,t)\rangle=(\beta\gamma/2\pi\,t)^{-1/2}\,\exp\Big({-\frac{\beta\gamma\,x^{2}}{2t}}\Big)\;.
\end{equation}

Our paper is organized as follows. In Section \ref{Section DPRM}, we recall the replica method and the mapping to the $\delta$-Bose gas. To provide some background, we explain the one-point generating function for the partition sum and the corresponding Fredholm determinant. The extension to two points is discussed and the long-time limit is obtained. The technical derivation is carried out in Sections \ref{Section 3} and  \ref{Section Fredholm G1} with supporting material in the Appendices.


\section{Main results}
\label{Section DPRM}
\setcounter{equation}{0}

\subsection{Scale invariance, stationarity}
To have a short hand the Brownian motion  average in (\ref{Z(x,t)}) is denoted by $\mathbb{E}^{\beta\gamma}_{(x,t)}$. To define it we first introduce the Gaussian average $\mathbb{E}^{\beta\gamma}$ with mean $0$ and covariance
\begin{equation}
\mathbb{E}^{\beta\gamma}\big(x(\tau)x(\tau')\big) = (\beta\gamma)^{-1} \mathrm{min}(\tau,\tau')\;,
\end{equation}
which corresponds to the Brownian motion starting at $0$, and set 
\begin{equation}
\mathbb{E}^{\beta\gamma}_{(x,t)}\big(\cdot\big) = \mathbb{E}^{\beta\gamma}\big(\cdot \,\delta(x(t)-x)\big)\;.
\end{equation}
In particular
\begin{equation}
\mathbb{E}^{\beta\gamma}_{(x,t)}\big(\delta(x(t)-x)\big) = (\beta\gamma/2\pi\,t)^{-1/2}\,\exp\Big({-\frac{\beta\gamma\,x^{2}}{2t}}\Big)\;.
\end{equation}
Let us denote the partition function (\ref{Z(x,t)}) by $Z_{\beta,\gamma,D}(x,t)$ to indicate explicitly the parameter dependence. From the scale invariance of white noise, $V(ax,bt)$ has the same distribution as 
$(ab)^{-1/2}V(x,t)$, and of the free directed polymer, $x(at)$ has the same distribution as $a^{1/2}x(t)$, one obtains
\begin{equation}
\label{scale invariance Z}
Z_{\beta,\gamma,D}(x,t)=\beta^{3}\gamma D\,Z_{1,1,1}\left(\beta^{3}\gamma D\,x,\beta^{5}\gamma D^{2}\,t\right)\;.
\end{equation}
Thus, it suffices to consider the case when all parameters are equal to one, and from now on we adopt the convention
\begin{equation}
\label{beta=gamma=D=1}
\beta=1\;,\quad\gamma=1\;,\quad D=1\;,\quad Z(x,t)=Z_{1,1,1}(x,t)\;,\quad \mathbb{E}^{\beta\gamma}_{(x,t)}=\mathbb{E}_{(x,t)}\;.
\end{equation}
The free energy (\ref{F[Z]}) has a systematic upward curvature of $x^{2}/2t$, compare with (\ref{<Z>}). However the distribution of $F(x,t)-x^{2}/2t$ is independent of the position $x$. This fact can be seen by performing a linear change of variables in the functional integration (\ref{Z(x,t)}) defining $Z(x,t)$. One obtains
\begin{eqnarray}
\fl\qquad \langle Z(x + a,t) \rangle^{-1} Z(x+a,t) &=& \langle Z(x,t) \rangle^{-1} \mathbb{E}_{(x,t)}\big( \exp\Big(-\int_{0}^{t}\rmd\tau V(x(\tau)+ a\tau /t,\tau)\Big)\big)\nonumber\\
\fl\qquad &=& \langle Z(x,t) \rangle^{-1}Z(x,t)\;,
\end{eqnarray}
where in the second equality we used that the white noise $V(x,\tau)$ is statistically translation invariant in the spatial argument. Hence the distribution of $Z(x,t)/\langle Z(x,t) \rangle$ is independent of $x$. More general, the stochastic process $x\mapsto Z(x,t)/\langle Z(x,t)\rangle$ is stationary in $x$, and correlation functions of the form
\begin{equation}
\big\langle f_1(Z(x_{1},t)/ \langle Z(x_1,t)\rangle)\ldots f_n(Z(x_{n},t)/ \langle Z(x_n,t)\rangle)\big\rangle\;
\end{equation}
are invariant under a global translation of all the arguments. In particular, the joint distribution of $Z(x_{1},t)/\langle Z(x_{1},t)\rangle$ and $Z(x_{2},t)/\langle Z(x_{2},t)\rangle$ depends only on the separation $x_{2}-x_{1}$.

For our computation of the two-point generating function it will be of advantage to keep the dependence on $x_1$, $x_2$ separately. The final result will confirm the parabolic free energy shift 
and the dependence on $x_2 - x_1$ only.

\subsection{Replicas}
The $n$-point correlation function of $Z(x,t)$ can be computed by introducing the replicas $x_{1}(\tau), \ldots, x_{n}(\tau)$, which are simply independent copies of the free directed polymer $x(\tau)$. More specifically, using the explicit form of the generating function for a Gaussian,
\begin{eqnarray}
\fl\quad \langle Z(x_{1},t)\ldots Z(x_{n},t)\rangle &=& \big\langle\prod\limits_{j=1}^{n} \mathbb{E}_{(x_j,t)}
\big(\exp\Big(-\int_{0}^{t} \rmd\tau V(x_{j}(\tau),\tau)\Big)\big)\big\rangle\nonumber\\
\fl\quad &=& \Big(\prod\limits_{j=1}^{n} \mathbb{E}_{(x_j,t)}\Big)\big(\exp\Big(\tfrac{1}{2}\int_{0}^t \rmd\tau \sum\limits_{i\neq j=1}^{n}\delta(x_{i}(\tau)-x_{j}(\tau))\Big)\big)\;,
\end{eqnarray}
where in the second line the average is over all replicas, the $j$-th replica starting at 0 and ending at $x_j$ at time $t$. It should be noted that the literal Gaussian average would also include the self-interaction term $-\tfrac{1}{2}\sum_{i=1}^n\delta(x_{i}-x_{i})$. Thus, for the moments of $Z$, the free energy renormalization needed to properly define (\ref{Z(x,t)}) simply corresponds to subtract the self-energy.

The Feynman-Kac formula implies that the $n$-point correlation function $\left\langle Z(x_{1},t)\ldots Z(x_{n},t)\right\rangle$ satisfies the imaginary time Schr\"odinger equation
\begin{equation}
\label{PDE<ZZZ>}
-\partial_{t}\left\langle Z(x_{1},t)\ldots Z(x_{n},t)\right\rangle=H_{n}\left\langle Z(x_{1},t)\ldots Z(x_{n},t)\right\rangle\;
\end{equation}
with the initial condition $Z(x_j,0) = \delta(x_j)$. Here $H_n$ is the Lieb-Liniger quantum Hamiltonian of $n$ particles on the line with an attractive $\delta$-interaction, 
\begin{equation}
\label{H Bose}
H_{n}=-\frac{1}{2}\sum_{i=1}^{n}(\partial_{x_{i}})^{2}-\frac{1}{2}\sum_{i\neq j=1}^{n}\delta(x_{i}-x_{j})\;
\end{equation}
\cite{LL63.1,MG64.1}. In this representation the free energy renormalization corresponds to the normal ordering of $H_n$. ``Solving'' (\ref{PDE<ZZZ>}), the $n$-point correlation function is given by
\begin{equation}
\label{<ZZZ>[Hn]}
\langle Z(x_{1},t)\ldots Z(x_{n},t)\rangle=\langle x_{1},\ldots,x_{n}|\rme^{-tH_{n}}|0\rangle\;.
\end{equation}
Here $|0\rangle$ is the state where all particles are at $0$ and $|x_1,\ldots,x_n\rangle$ is the one where the $j$-th particle is at $x_j$. Since $|0\rangle$ is symmetric under the exchange of particle labels, one can symmetrize in the final state. Thus the propagator $\rme^{-tH_{n}}$ is needed only in the symmetric sector and the replicas are expressed by the attractive $\delta$-Bose gas on the line.

As first shown by Mc Guire \cite{MG64.1}, the ground state energy $E_{0}(n)$ of $H_{n}$ is given by
\begin{equation}
E_{0}(n)=-\frac{1}{24}(n^{3}-n)\;.
\end{equation}
The term linear in $n$ translates to the free energy shift $t/24$, which in fact equals the bulk free energy per unit time. For later use, we introduce the parameter
\begin{equation}
\label{alpha}
\alpha=(t/2)^{1/3}\;.
\end{equation}
In the lowest order approximation, ignoring all excited states,
\begin{equation}
\label{approx groundstate}
\langle0|\rme^{-tH_{n}}|0\rangle\approx \rme^{-tE_{0}(n)}\;.
\end{equation}
The cubic term of $E_{0}(n)$ leads to the decay of the left tail of $F(0,t)$ as
\begin{equation}
\label{distrib tail F}
\mathrm{Prob}\left(F(0,t)-\tfrac{1}{24}t\leq u\right)\approx\exp\left[-\tfrac{4}{3}\left(\alpha|u|\right)^{3/2}\right]\;
\end{equation}
for $u\to-\infty$ and $t^{-1/3}|u|=\O{1}$ \cite{K87.1,BO90.1}, see also \cite{SS10.2}. This result confirms that the fluctuations of the free energy are of order $t^{1/3}$. In fact, the tail behaviour (\ref{distrib tail F}) agrees with the exact tail, see \cite{PS02.2}.

\subsection{One-point generating function of the free energy}
To go beyond (\ref{approx groundstate}), one needs the excited states of the attractive $\delta$-Bose gas. They can be computed from the Bethe ansatz as has been recently worked out in great detail by Dotsenko and Klumov \cite{DK10.1,D10.2}. In brackets, we remark that on a ring, the complex momenta are solutions of the nonlinear Bethe equations on which little information is available. Already to determine the ground state energy requires ingenious computations \cite{BD00.1}. A corresponding situation has been found for the asymmetric simple exclusion process. On a ring, while the ground state and the large deviations for the current have been extensively investigated \cite{DL98.1,LK99.1,PM09.1,P10.1}, the Bethe equations for excited states have been analyzed only partially \cite{GS92.1,K95.1,GM05.1,GM05.2}. In contrast, for the infinite lattice there is a reasonably concise formula for the transition probability with any given number of particles \cite{TW08.1}.

Using the complete eigenfunction expansion for $H_n$, a particular generating function can be expressed as a Fredholm determinant \cite{CLDR10.1,D10.1,D10.2}. More precisely, we define
\begin{equation}
\label{G[F] 1 pt}
G(s;x,t)=\left\langle\exp\left(-\rme^{-s}\rme^{-F(x,t)}\right)\right\rangle\;.
\end{equation}
Then $G$ is equal to the Fredholm determinant
\begin{equation}
\label{G Fredholm 1pt}
G\left(s-\tfrac{1}{24}t -\tfrac{1}{2t}x^{2};\,x,\,t\right)=\det(1-M)\;.
\end{equation}
The operator $M$ does not depend on the position $x$ and depends on time $t$ only through the parameter $\alpha$ defined in (\ref{alpha}). $M$ acts on $L^{2}(\mathbb{R})$ and has the integral kernel
\begin{equation}
\label{kernel M}
\langle u|M|v\rangle=\frac{\rme^{\alpha u-s}}{1+\rme^{\alpha u-s}}\langle u|K|v\rangle\;.
\end{equation}
Here $K$ is the Airy operator with integral kernel
\begin{equation}
\langle u|K|v\rangle=\int_{0}^{\infty}\rmd z\,\Ai(u+z)\Ai(z+v)\;,
\end{equation}
called \textit{Airy kernel}. $K$ is related to the Airy Hamiltonian
\begin{equation}
H = -(\partial_u)^{2} +u\;,
\end{equation}
as $K$ projects onto all negative eigenstates of $H$. In particular $K ^* = K$ and $K^{2} = K$, see \ref{Appendix Airy} for details. One easily checks that $\tr|M|<\infty$. Hence the Fredholm determinant in (\ref{G Fredholm 1pt}) is well defined.

The mathematical status of (\ref{G Fredholm 1pt}) is somewhat tricky. One cannot simply verify (\ref{G Fredholm 1pt}) as the solution of some equation. The derivation relies on the replica method. Since $\log\langle Z(x,t)^{n}\rangle\approx n^{3}$, the moments do not uniquely determine the distribution of $Z(x,t)$. To derive (\ref{G Fredholm 1pt}), one is forced to work with divergent series and has to make a reasonable choice for the analytic extension of $\langle Z^{n}\rangle$, $n\in\mathbb{N}$, to the complex plane \cite{CLDR10.1,D10.1}. However, the generating function $G(s;x,t)$ fixes the distribution of $F(x,t)$, which is known by other means \cite{ACQ10,SS10.1} Thus \textit{a posteriori} one can verify directly that (\ref{G Fredholm 1pt}) is a valid identity.

Eq. (\ref{G Fredholm 1pt}) together with (\ref{kernel M}) establishes that $F$ is of order $\alpha$. Rescaling $s$ as $\alpha a$ and taking $\alpha \to \infty$, we note that the right side of (\ref{G[F] 1 pt}) and the multiplicative prefactor in (\ref{kernel M}) both converge to a step function. Hence in the long time limit one obtains
\begin{equation}
\label{distrib 1pt F}
\fl\qquad\qquad\lim_{t\to\infty}\mathrm{Prob}\left(F(x,t)-\tfrac{1}{24}t-\tfrac{1}{2t}x^{2}<-\alpha a\right)=\det\big(1-P_{a}KP_{a}\big)=F_{2}(a)\;,
\end{equation}
where $P_{a}$ projects on $[a,\infty)$. In the long time limit, the free energy fluctuations are thus of order $t^{1/3}$. The function $F_{2}$ is the celebrated Tracy-Widom distribution function \cite{TW94.1}, which first appeared as the distribution for the maximal eigenvalue of large random Hermitian matrices in the Gaussian unitary ensemble (GUE).

By a more sophisticated argument \cite{CLDR10.1} one also deduces the finite time probability density from (\ref{G Fredholm 1pt}) with the result
\begin{equation}
\label{distrib 1pt F(x,t)}
\fl\qquad\qquad\mathrm{Prob}\left(F(x,t)-\tfrac{1}{24} t-\tfrac{1}{2t}x^{2}<-\alpha s\right)=\int_{-\infty}^{\infty}\rmd u\,\exp\left(-\rme^{\alpha(s-u)}\right)g_{t}(u)\;,
\end{equation}
where
\begin{equation}
\label{gt 1pt}
g_{t}(u)=\det\big(1-P_{u}(B_{t}-P_{\Ai})P_{u}\big)-\det\big(1-P_{u}B_{t}P_{u}\big)\;.
\end{equation}
The operators $P_{\Ai}$ and $B_{t}$ are defined respectively by
\begin{equation}
\langle z|P_{\Ai}|z'\rangle=\Ai(z)\Ai(z')\;,
\end{equation}
and
\begin{equation}
\langle z|B_{t}|z'\rangle=\int_{-\infty}^{\infty}\rmd v\,\frac{1}{1-\rme^{\alpha v}}\Ai(z+v)\Ai(v+z')\;.
\end{equation}
There is another model for which the one-point generating function is available \cite{Con10}: one replaces the Brownian motion by a continuous time random walk on the lattice $\mathbb{Z}$ with forward jumps only. The white noise is correspondingly discretized in the spatial direction. The partition function is $Z_{\beta}(N,t)$ with the polymer starting at $0$ and ending at $(N,t)$, $N\geq 0$. Note that this model has less scale invariance than our case because of the lattice. The long time behavior of the free energy remains to be studied. 

For the directed polymer at zero temperature,  results are available for a lattice discretization in case the random potential has either a one-sided exponential or geometric distribution \cite{J00.1}.

\subsection{Two-point generating function and long time limit}
Our novel contribution is the extension of (\ref{G Fredholm 1pt}) to two reference points by using the replica method. In analogy to (\ref{G[F] 1 pt}), let us define the generating function
\begin{equation}
\label{G[F] 2pt}
G(s_{1},s_{2};x_{1},x_{2},t) = \left\langle\exp\left(-\rme^{-s_{1}}\rme^{-F(x_{1},t)}-\rme^{-s_{2}}\rme^{-F(x_{2},t)}\right)\right\rangle\;.
\end{equation}
From the mapping of the directed polymer to the $\delta$-Bose gas, the generating function $G$ can be expanded in a sum over eigenstates of the Lieb-Liniger Hamiltonian (\ref{H Bose}). Using results from \cite{D10.2}, under a natural factorization assumption the sum over the eigenstates can be written again as a Fredholm determinant. However even then it is difficult to extract any useful information from this representation. By a sequence of miraculous transformations we arrive at an alternative expression for the Fredholm determinant, which turns out to be rather similar to (\ref{kernel M}) in structure and from which the long time limit can be read off easily.  To distinguish from the true generating function, we introduce the sharp superscript ${}^\sharp$ for generating functions with factorization assumption.

Let us first define the function $\Phi$ by
\begin{equation}
\label{Phi}
\Phi(u,v)=\frac{\rme^{u}+\rme^{v}}{1+\rme^{u}+\rme^{v}}
\end{equation}
and the operator $Q$ through the kernel
\begin{equation}
\label{Q}
\langle u|Q|v\rangle = \Phi(\alpha u-s_{1},\;\alpha v-s_{2})\langle u|\rme^{-(2\alpha^{2})^{-1}|x_1 - x_2|H}|v\rangle\;.
\end{equation}
Then
\begin{equation}
\label{G Fredholm 2pt}
\fl\quad G^\sharp(s_{1}-\tfrac{1}{24}t-\tfrac{1}{2t}x_{1}^{2},\,s_{2}-\tfrac{1}{24}t-\tfrac{1}{2t}x_{2}^{2};\;x_{1},x_{2},\;t)
= \det\big(1-Q\,\rme^{(2\alpha^{2})^{-1}|x_1 - x_2|H}K\big)\;.
\end{equation}
Recall that $\langle u|\rme^{-(2\alpha^{2})^{-1}|x_{1}-x_{2}|H}|v\rangle$ is the propagator of the Airy Hamiltonian and note that
\begin{equation}
\label{HK[Ai]}
\fl\qquad\qquad\langle u |\rme^{(2\alpha^{2})^{-1}|x_1 - x_2|H}K |v\rangle = \int_0^\infty\rmd z\,\rme^{-(2\alpha^{2})^{-1}|x_1 - x_2|z}\,\Ai(u+z)\Ai(z+v)\;.
\end{equation}
In particular, $\rme^{(2\alpha^{2})^{-1}|x_{1}-x_{2}|H}K$ is a bounded operator.

From (\ref{G Fredholm 2pt}) it is obvious that for long times $F$ scales as $\alpha \sim t^{1/3}$ and the two-point distribution 
has a non-degenerate limit only if $x$ scales as $\alpha^{2} \sim t^{2/3}$. So let us substitute $s_{1}$ by $\alpha a$, $s_{2}$ by $\alpha b$, and introduce
\begin{equation}
|x_1 - x_2| = 2\alpha^{2} y\;.
\end{equation}
Then as $t \to \infty$ the left hand side of (\ref{G Fredholm 2pt}) converges to the characteristic function of the rectangle
$[-\alpha a, \infty)\times [-\alpha b , \infty)$, while the operator $Q$ converges to $P_a \rme^{-|y|H} + \rme^{-|y|H}P_b - P_a \rme^{-|y|H}P_b $. With $y$, $a$ and $b$ held fixed, we arrive at
\begin{eqnarray}
\label{distrib 2pt F}
\fl&&\qquad\qquad\lim_{t\to\infty}\mathrm{Prob}\left(F(x_{1},t)-\tfrac{1}{24}t-\tfrac{1}{2t}x_{1}^{2}>-\alpha a,\,F(x_{2},t)-\tfrac{1}{24}t-\tfrac{1}{2t}x_{2}^{2}>-\alpha b\right)\nonumber\\
\fl&&\qquad\qquad\qquad=F_2(a,b;y)\;,
\end{eqnarray}
where the function $F_{2}(a,b;y)$ is given in terms of the Fredholm determinant
\begin{equation}
\label{F2(a,b)}
F_{2}(a,b;y)=\det\left(1-\left(P_{a}+\rme^{-|y|H}P_{b}\rme^{|y|H}-P_{a}\rme^{-|y|H}P_{b}\rme^{|y|H}\right)K\right)\;.
\end{equation}
The function $F_{2}(a,b;y)$ is two-point distribution of the Airy process in a form written down first in Eq. (5.8) of \cite{PS02.2}, where the two-point distribution function of the height in the polynuclear growth droplet model was studied. For the single step growth model the corresponding result was achieved in \cite{J03}. Thus apparently  the factorization becomes exact in the long time limit.

As shown in \cite{PS02.2} the function $F_{2}(a,b;y)$ can also be expressed as the Fredholm determinant 
of a $2\times 2$ operator kernel,
\begin{equation}
\label{F2(a,b) 2*2 matrix kernel}
F_{2}(a,b;y)=\det\Big[1-\left(\begin{array}{cc}P_{a}&0\\0&P_{b}\end{array}\right)\left(\begin{array}{cc}K&\rme^{-|y|H}(K-1)\\\rme^{|y|H}K&K\end{array}\right)\Big]\;.
\end{equation}
This form arises naturally when one studies the top most line in Dyson's Brownian motion. In the large $N$ limit it converges to the Airy process, which can be viewed as the top line of an underlying extended determinantal random field. In particular, the $n$-point distribution is defined most directly through an operator with a $n\times n$ matrix structure.
While we were searching for a corresponding matrix structure, it came as a surprise that the replica route apparently prefers the expression (\ref{F2(a,b)}).

From (\ref{F2(a,b) 2*2 matrix kernel}) one can read of properties of $F_{2}(a,b;y)$. Obviously it is symmetric in $a$ and $b$. In the limits $y\to0$ and $y\to\infty$, the expression of $F_{2}(a,b;y)$ simplifies and one obtains
\begin{equation}
\fl\qquad\lim_{y\to0}F_{2}(a,b;y)=F_{2}(\min(a,b))\qquad\text{and}\qquad\lim_{y\to\infty}F_{2}(a,b;y)=F_{2}(a)F_{2}(b)\;,
\end{equation}
with $F_{2}$ of a single argument denoting the Tracy-Widom distribution (\ref{distrib 1pt F}). Using (\ref{F2(a,b)}) and $[H,K] = 0$ one finds 
\begin{equation}
\lim_{b\to-\infty}F_{2}(a,b;y)=0\qquad\text{and}\qquad\lim_{b\to\infty}F_{2}(a,b;y)=F_{2}(a)\;.
\end{equation}
Since $F_{2}(a,b;y)$ is a distribution function, see (\ref{distrib 2pt F}), it is an increasing function of $a,b$. It tends to $0$ when $a, b \to -\infty$, and to $1$ when $a,b \to \infty$.


\section{Replica Summation}
\label{Section 3}
\setcounter{equation}{0}

\subsection{Two-point generating function}
We start from the expression (\ref{G[F] 2pt}) of the generating function $G(s_{1},s_{2};x_{1},x_{2},t)$
and define 
\begin{eqnarray}
G_{1}&=& G\left(s_{1}-\tfrac{1}{24}t,s_{2}-\tfrac{1}{24}t;x_{1},x_{2},t\right)\nonumber\\
&=&\left\langle\exp\left(-\rme^{t/24}\rme^{-s_{1}}Z(x_{1},t)-\rme^{t/24}\rme^{-s_{2}}Z(x_{2},t)\right)\right\rangle\;,
\end{eqnarray}
subtracting the linear order of the free energy. Expanding the exponential, $G_{1}$ is written in terms of $n$-point correlations of the partition function. After decomposing as 
\begin{equation}
\label{X x}
X = \frac{x_{1} + x_{2}}{2}\;,\quad x = \frac{x_{2} - x_{1}}{2}\;,
\end{equation}
we obtain
\begin{eqnarray}
\label{G1[<ZZZ>]}
G_{1}&=&1+\sum_{N=1}^{\infty}\frac{(-1)^{N}\rme^{tN/24}}{N!}\sum_{\sigma_{1},\ldots,\sigma_{N}=\pm1}\Big(\prod_{i=1}^{N}\rme^{-\frac{1}{2}\left[(1-\sigma_{i})s_{1}+(1+\sigma_{i})s_{2}\right]}\Big)\nonumber\\
&&\times\left\langle Z(X+\sigma_{1}x,t)\ldots Z(X+\sigma_{N}x,t)\right\rangle\;.
\end{eqnarray}
We use the expression (\ref{<ZZZ>[Hn]}) for the $N$-point correlations in terms of the $\delta$-Bose gas. These matrix elements can be expanded as a sum over the orthornormal basis of eigenstates, $\psi_{r}$, of the Lieb-Liniger Hamiltonian $H_{N}$ with the result
\begin{equation}
\label{<ZZZ>[psi]}
\left\langle Z(x_{1},t)\ldots Z(x_{N},t)\right\rangle=\sum_{r}\rme^{-tE_{r}}\langle x_{1},\ldots,x_{N}|\psi_{r}\rangle\langle\psi_{r}|0\rangle\;,
\end{equation}
where $E_{r}$ is the corresponding energy eigenvalue. Combining (\ref{G1[<ZZZ>]}) and (\ref{<ZZZ>[psi]}), one obtains the following expression for the generating function $G_{1}$,
\begin{eqnarray}
\label{G1[psi]}
&&\fl\qquad\qquad G_{1}=1+\sum_{N=1}^{\infty}\frac{(-1)^{N}\rme^{tN/24}}{N!}\sum_{r}\rme^{-tE_{r}}|\psi_{r}(0,\ldots,0)|^{2}\!\nonumber\\
&&\fl\qquad\qquad\times\sum_{\sigma_{1},\ldots,\sigma_{N}=\pm1}\Big(\prod_{i=1}^{N}\rme^{-\frac{1}{2}\left[(1-\sigma_{i})s_{1}+(1+\sigma_{i})s_{2}\right]}\Big)\frac{\psi_{r}(X+\sigma_{1}x,\ldots,X+\sigma_{N}x)}{\psi_{r}(0,\ldots,0)}\;.
\end{eqnarray}

\subsection{Summation over the eigenstates of the $\delta$-Bose gas}
We have to perform the summation over the eigenstates in the expression (\ref{G1[psi]}) of the generating function $G_{1}$, for which we follow \cite{D10.2} with some improvements. To ease the comparison our notation will be as close as possible to the one in \cite{D10.2}.

The eigenfunctions of the $N$-particle attractive $\delta$-Bose gas are labelled by the complex wave numbers $\xi_{a}$, $a=1,\ldots,N$. To write them down, we pick positive integers $n_{\alpha}$, $\alpha=1,\ldots,M\leq N$  such that
\begin{equation}
\sum_{\alpha=1}^{M}n_{\alpha}=N\;.
\end{equation}
We also introduce the
running indices $r_{\alpha}=1,\ldots,n_{\alpha}$ and set
\begin{equation}
n^{(\alpha)}=\sum_{\beta=1}^{\alpha}n_{\beta}\;, \quad n^{(0)} = 0\;, \quad n^{(M)} = N\;.
\end{equation}
Then, for arbitrary vectors $\boldsymbol{q}=(q_{1},\ldots,q_{M}) \in \mathbb{R}^M$ and $\boldsymbol{n}=(n_{1},\ldots,n_{M})$, one has
\begin{equation}
\xi_{a}=q_{\alpha}-\frac{\rmi}{2}(n_{\alpha}+1-2r_{\alpha})\quad\mathrm{for} \quad a = n^{(\alpha-1)}+r_{\alpha}\;,
\end{equation}
with $\alpha=1,\ldots,M$. According to \cite{D10.2}, Eq. (B.21), the eigenfunction with labels $\boldsymbol{q}, \boldsymbol{n}$ is given by
\begin{eqnarray}
\label{psi q n M}
&&\fl\qquad\psi_{\boldsymbol{q},\boldsymbol{n}}^{(M)}(x_1,...,x_N)\nonumber\\
&&\fl\qquad=C_{\boldsymbol{q},\boldsymbol{n}}^{(M)}\sum_{p\in\mathcal{P}}\sgn (p)\prod_{1\leq a<b\leq N}\big(\xi_{p(a)}-\xi_{p(b)}+\rmi\sgn(x_{a}-x_{b})\big)\;\exp\Big[\rmi\sum_{c=1}^{N}\xi_{p(c)}x_{c}\Big]\;.
\end{eqnarray}
The normalization constant $C_{\boldsymbol{q},\boldsymbol{n}}^{(M)}$ is computed in \cite{D10.2}. 
The sum is over the set $\mathcal{P}$ of all  $N$-long permutations and $\sgn(p)$ is the signature of the permutation $p$. $\psi_{\boldsymbol{q},\boldsymbol{n}}^{(M)}$ is continuous and symmetric in $x_{1}$, \ldots, $x_{N}$. Each eigenstate of the $\delta$-Bose gas is made up of $M$ clusters, where the $\alpha$'s cluster consists of $n_{\alpha}$ bound particles and has center of mass momentum $q_{\alpha}$.

In the following transformation we invoke a combinatorial identity (we owe the proof to P. Di Francesco).
\begin{lemma}
\label{Lemma Vandermonde}
Let $f(a,b)$ be arbitrary complex coefficients and let
\begin{equation}
\label{D Vandermonde 1}
D(\xi_{1},\ldots,\xi_{N})=\sum_{p\in\mathcal{P}}\sgn(p)\prod_{1\leq a<b\leq N}\big(\xi_{p(a)}-\xi_{p(b)}+f(a,b)\big)\;.
\end{equation}
Then $D$ equals the Vandermonde determinant,
\begin{equation}
\label{D Vandermonde 2}
D(\xi_{1},\ldots,\xi_{N})=N!\prod_{1\leq a<b\leq N}(\xi_{a}-\xi_{b})\;.
\end{equation}
\end{lemma}
\begin{proof}
$D$ is a polynomial of total degree $N(N-1)/2$ with leading coefficient $N!\prod_{1\leq a<b\leq N}(\xi_{a}-\xi_{b})$. Hence, to prove (\ref{D Vandermonde 2}) one only has to establish that $D$ is antisymmetric, since any antisymmetric polynomial is divisible by the Vandermonde determinant.

It suffices to study the interchange of a specific pair, say $\xi_{1}$ and $\xi_{2}$. Let $p$ be some permutation such that $p(c)=1$, $p(d)=2$ and $c<d$. Then the product in (\ref{D Vandermonde 1}) decomposes into a factor $C(p)$ independent of $\xi_{1},\xi_{2}$ and a second factor as
\begin{eqnarray}
\label{eq proof lemma}
&&\fl\quad\sgn(p)C(p)(\xi_{1}-\xi_{2}+f(c,d))\prod_{a<c}\big(\xi_{p(a)}-\xi_{1}+f(a,c)\big)\!\!\!\!\prod_{\substack{a<d\\(a\neq c)}}\!\!\!\!\big(\xi_{p(a)}-\xi_{2}+f(a,d)\big)\nonumber\\
&&\fl\quad\qquad\qquad\qquad\qquad\qquad\times\!\!\!\!\prod_{\substack{b>c\\(b\neq d)}}\!\!\!\!\big(\xi_{1}-\xi_{p(b)}+f(c,b)\big)\prod_{b>d}\big(\xi_{2}-\xi_{p(b)}+f(d,b)\big)\;.
\end{eqnarray}
Let $\tilde{p}$ be the permutation with $1$ and $2$ interchanged. Then $C(\tilde{p})=C(p)$, $\sgn(\tilde{p})=-\sgn(p)$, $\tilde{p}(c)=2$, $\tilde{p}(d)=1$ and the decomposition as in (\ref{eq proof lemma}) reads
\begin{eqnarray}
&&\fl\quad-\sgn(p)C(p)(\xi_{2}-\xi_{1}+f(c,d))\prod_{a<c}\big(\xi_{p(a)}-\xi_{2}+f(a,c)\big)\!\!\!\!\prod_{\substack{a<d\\(a\neq c)}}\!\!\!\!\big(\xi_{p(a)}-\xi_{1}+f(a,d)\big)\nonumber\\
&&\fl\quad\qquad\qquad\qquad\qquad\qquad\times\!\!\!\!\prod_{\substack{b>c\\(b\neq d)}}\!\!\!\!\big(\xi_{2}-\xi_{p(b)}+f(c,b)\big)\prod_{b>d}\big(\xi_{1}-\xi_{p(b)}+f(d,b)\big)\;.
\end{eqnarray}
It is now obvious that the sum of the two terms is antisymmetric in $\xi_{1}$ and $\xi_{2}$, and hence $D(\xi_{1},\xi_{2},\ldots,\xi_{N})=-D(\xi_{2},\xi_{1},\ldots,\xi_{N})$.
\end{proof}
We first compute $\psi_{\boldsymbol{q},\boldsymbol{n}}^{(M)}(0)$ by taking the limit $\varepsilon\to0$ of $x_{j}=\varepsilon j$. Then $\sgn(x_{a}-x_{b})=-1$ for all $a<b$ and by Lemma \ref{Lemma Vandermonde}
\begin{equation}\label{III.14}
\psi_{\boldsymbol{q},\boldsymbol{n}}^{(M)}(0)=C_{\boldsymbol{q},\boldsymbol{n}}^{(M)}N!\prod_{1\leq a<b\leq N}(\xi_{a}-\xi_{b})\;.
\end{equation}
Secondly we have to evaluate (\ref{psi q n M}) at $x_j = X + x\sigma_j$ and to perform the sum over all  ``spin" configurations $\boldsymbol{\sigma} = \{\sigma_1,...,\sigma_N\}$, for which purpose we use that (\ref{psi q n M}) can be written more compactly by using the special structure of the complex wave numbers, see  \cite{D10.2}, Sect. B.2. We introduce the cluster counting function $\alpha: [1,...,N]  \to 
[1,...,M]$ by
\begin{equation}
\alpha(a) = \beta \quad \mathrm{for} \quad n^{(\beta -1)} < a \leq n^{(\beta)}\;,\quad \beta = 1,...,M\;,
\end{equation}
and the $\beta$-th cluster by
\begin{equation}
\Omega_{\beta}(p)=\{a\,|\, \alpha(p(a))={\beta}\}\;.
\end{equation}
Then, working out the derivatives in \cite{D10.2}, Eq. (B.28),
\begin{eqnarray}
&&\fl\quad\psi_{\boldsymbol{q},\boldsymbol{n}}^{(M)}(x_1,...,x_N) = C_{\boldsymbol{q},\boldsymbol{n}}^{(M)} {\sum_{p\in\mathcal{P}}}'\sgn(p)\!\!\!\!\prod_{\substack{1\leq a ,b\leq N\\ \alpha(p(a)) \neq \alpha(p(b))}}\!\!\!\!\big(q_{\alpha(p(a))} - q_{\alpha(p(b))} + \rmi \eta(x_a,x_b)\big)\nonumber\\    
&&\fl\quad\qquad\qquad\qquad\qquad\qquad\times \exp \Big[\rmi \sum_{\alpha = 1}^M q_{\alpha} \sum_{c \in \Omega_{\alpha}(p) }x_c - \frac{1}{4} \sum_{\alpha = 1}^M \sum_{c,c' \in \Omega_{\alpha}(p)}|x_c - x_{c'}| \Big]\;.
\end{eqnarray}
Here the sum over permutations is understood modulo permutations inside each cluster, as indicated by $'$,
and $\eta$ is defined by
\begin{equation}
\fl\qquad\eta(x_a,x_b) = \sgn(x_a - x_b) + \frac{1}{2}\Big( \!\!\!\!\!\! \sum_{\substack{c \in \Omega_{\alpha}(p)\\c \neq a \in \Omega_{\alpha}(p)}} \!\!\!\!\!\!\!\! \sgn(x_a - x_c)  - \!\!\!\!\!\! \sum_{\substack{c \in \Omega_{\alpha}(p)\\c \neq b \in \Omega_{\alpha}(p)}} \!\!\!\!\!\!\!\! \sgn(x_b - x_c) \Big)\;.
\end{equation} 
We spread the particle configuration as
\begin{equation}
x_a = X + x\sigma_a +\varepsilon a\;.
\end{equation}
Then
\begin{eqnarray}
\label{III.20}
&&\fl\qquad\sum_{\boldsymbol{\sigma}}\Big(\prod_{i=1}^{N}\rme^{-\frac{1}{2}\left[(1-\sigma_{i})s_{1}+(1+\sigma_{i})s_{2}\right]}\Big)\psi_{\boldsymbol{q},\boldsymbol{n}}^{(M)}(X+\sigma_{1}x +\varepsilon,\ldots,X+\sigma_{N}x +N\varepsilon)\nonumber\\
&&\fl\qquad=C_{\boldsymbol{q},\boldsymbol{n}}^{(M)}\sum_{\boldsymbol{\sigma}}{\sum_{p\in\mathcal{P}}}' \sgn(p) \!\!\!\!\!\! \prod_{\substack{1\leq a ,b\leq N\\ \alpha(p(a)) \neq \alpha(p(b))}} \!\!\!\!\!\! \big(q_{\alpha(p(a))} - q_{\alpha(p(b))} + \rmi \eta(x_a,x_b)\big)\;\rme^{\phi(\boldsymbol{\sigma},p)}\;.
\end{eqnarray}
The phase $\phi(\boldsymbol{\sigma},p)$ is given by
\begin{eqnarray}
&&\fl\qquad\phi(\boldsymbol{\sigma},p) = \sum_{\alpha = 1}^M\Big(-\tfrac{1}{2}(s_1 + s_2)n_{\alpha} +\tfrac{1}{2}(s_1 - s_2)m_{\alpha}(\boldsymbol{\sigma},p)  
\nonumber\\
&&\fl\qquad\qquad\qquad\qquad\qquad + \rmi q_{\alpha}(X n_{\alpha} +x m_{\alpha}(\boldsymbol{\sigma},p))-\tfrac{1}{2}|x|(n_{\alpha}^{2} - m_{\alpha}(\boldsymbol{\sigma},p)^{2} )\Big)\;,
\end{eqnarray}
where we introduced
\begin{equation}
m_{\alpha}(\boldsymbol{\sigma},p) =  \sum_{c\in\Omega_{\alpha}(p)}\sigma_c\;.
\end{equation}
Inserting in (\ref{III.20}) one arrives at
\begin{eqnarray}
\label{III.22}
&&\fl\hspace{30pt}\sum_{\boldsymbol{\sigma}}\Big(\prod_{i=1}^{N}\rme^{-\frac{1}{2}\left[(1-\sigma_{i})s_{1}+(1+\sigma_{i})s_{2}\right]}\Big)\psi_{\boldsymbol{q},\boldsymbol{n}}^{(M)}(X+\sigma_{1}x +\varepsilon,\ldots,X+\sigma_{N}x +N\varepsilon)
\nonumber\\
&&\fl\hspace{30pt}=C_{\boldsymbol{q},\boldsymbol{n}}^{(M)}\sum_{\boldsymbol{\sigma}}{\sum_{p\in\mathcal{P}}}'
\sgn(p) \prod_{\substack{1\leq a ,b\leq N\\ \alpha(p(a)) \neq \alpha(p(b))}}\big(q_{\alpha(p(a))} - q_{\alpha(p(b))} +
\rmi \eta(\sigma_a +\varepsilon a,\sigma_b + \varepsilon b)\big)
\nonumber\\
&&\fl\hspace{50pt} \times\prod_{\alpha=1}^{M}\prod_{c\in\Omega_{\alpha}(p)}\exp\big[-\tfrac{1}{2}(s_1 + s_2)n_{\alpha} +\tfrac{1}{2}(s_1 - s_2)m_{\alpha}(\boldsymbol{\sigma},p) 
\nonumber\\
&&\fl\hspace{140pt} 
+ \rmi q_{\alpha}(X n_{\alpha} +x m_{\alpha}(\boldsymbol{\sigma},p)) -\tfrac{1}{2}|x|(n_{\alpha}^{2} - m_{\alpha}(\boldsymbol{\sigma},p)^{2} )\big]\;.
\end{eqnarray}
Let us shorthand the right hand side of Eq. (\ref{III.22}) as 
\begin{equation}
\sum_{\boldsymbol{\sigma}}\sum_{p\in\mathcal{P}} f_1(\boldsymbol{\sigma},p)f_2(\boldsymbol{\sigma},p)\;.
\end{equation}
By Lemma 1
\begin{equation}
\sum_{p\in\mathcal{P}} f_1(\boldsymbol{\sigma},p) = \psi_{\boldsymbol{q},\boldsymbol{n}}^{(M)}(0)
\end{equation}
not depending on $\boldsymbol{\sigma}$ and, since $f_2(\boldsymbol{\sigma},p)$ depends on $\boldsymbol{\sigma}$ only through the $m_{\alpha}(\boldsymbol{\sigma},p)$'s, correspondingly
\begin{equation}
\sum_{\boldsymbol{\sigma}}f_2(\boldsymbol{\sigma},p) = \tilde{c}
\end{equation}
with $\tilde{c}$ not depending on $p$. Unfortunately we could not discover any further simplification. To proceed anyhow a natural step is to factorize (\ref{III.22}) either with respect to $p$ or with respect to $\boldsymbol{\sigma}$, both leading  in approximation to 
\begin{eqnarray}
\label{III.23}
&&\fl\hspace{10pt}\sum_{\boldsymbol{\sigma}}\Big(\prod_{i=1}^{N}\rme^{-\frac{1}{2}\left[(1-\sigma_{i})s_{1}+(1+\sigma_{i})s_{2}\right]}\Big)\big(\psi_{\boldsymbol{q},\boldsymbol{n}}^{(M)}(X+\sigma_{1}x +\varepsilon,\ldots,X+\sigma_{N}x +N\varepsilon)/
\psi_{\boldsymbol{q},\boldsymbol{n}}^{(M)}(0)\big)\nonumber\\
&&\fl\hspace{40pt}\simeq \sum_{\boldsymbol{\sigma}}\prod_{\alpha = 1}^{M}\exp\big[-\tfrac{1}{2}(s_1 + s_2)n_{\alpha} +\tfrac{1}{2}(s_1 - s_2)m_{\alpha}(\boldsymbol{\sigma})\nonumber\\
&&\fl\hspace{90pt} + \rmi q_{\alpha}(X n_{\alpha} +x m_{\alpha}(\boldsymbol{\sigma})) -\tfrac{1}{2}|x|(n_{\alpha}^{2} - m_{\alpha}(\boldsymbol{\sigma})^{2} )
\big]\;,
\end{eqnarray}
where
\begin{equation}
m_{\alpha}(\boldsymbol{\sigma}) =  \sum_{r_{\alpha} = 1}^{n_{\alpha}}\sigma_{n^{(\alpha -1)}+r_{\alpha}}\;.
\end{equation}
For small cluster sizes we checked that (\ref{III.23}) is indeed not a strict equality.


\section{Two-point generating function}
\label{Section Fredholm G1}
\setcounter{equation}{0}

\subsection{Linearization and ``spin" summation}
We linearize the terms quadratic in $n_{\alpha}$ and in $m_{\alpha}(\boldsymbol{\sigma})$ in the exponential 
of (\ref{III.23}), so to be able to perform the summation over the $\sigma_{i}$'s and the $n_{\alpha}$'s. For this purpose we use the identity
\begin{equation}
\rme^{au+bv+cuv}=\rme^{c\partial_{a}\partial_{b}}\rme^{au+bv}\;,
\end{equation}
which can be checked by expanding both sides of the equation as a formal power series in $c$, and obtain
\begin{eqnarray}
&&\fl\quad\sum_{\boldsymbol{\sigma}}\Big(\prod_{i=1}^{N}\rme^{-\frac{1}{2}\left[(1-\sigma_{i})s_{1}+(1+\sigma_{i})s_{2}\right]}\Big)\psi_{\boldsymbol{q},\boldsymbol{n}}^{(M)}(X+\sigma_{1}x,\ldots,X+\sigma_{N}x)/\psi_{\boldsymbol{q},\boldsymbol{n}}^{(M)}(0)\nonumber\\
&&\fl\quad\simeq \sum_{\boldsymbol{\sigma}}\prod_{\alpha = 1}^{M}\rme^{-x\partial_{1}\partial_{2}}\exp\big[-\tfrac{s_{1}}{2}(n_{\alpha}-m_{\alpha}(\boldsymbol{\sigma}))-\tfrac{s_{2}}{2}(n_{\alpha}+m_{\alpha}(\boldsymbol{\sigma}))+\rmi q_{\alpha}(Xn_{\alpha}+xm_{\alpha}(\boldsymbol{\sigma}))\big]\;.\nonumber\\
&&
\end{eqnarray}
Here we introduced the convention
\begin{equation}
\partial_1 = \partial_{s_1}\;,\quad\partial_2 = \partial_{s_2}\;,
\end{equation}
which will be used onwards. We note that the exponential inside the product over $\alpha$ depends only on the $\sigma_{a}$ with  $n^{(\alpha-1)}  < a \leq n^{(\alpha)}$. Thus, the summation over the $\sigma_{a}$ can be performed independently for each $\alpha$. Recalling (\ref{X x}), we find
\begin{eqnarray}
\label{IV.2}
&&\fl\qquad\sum_{\boldsymbol{\sigma}}\prod_{\alpha = 1}^{M}\rme^{-x\partial_{1}\partial_{2}}\exp\big[-\tfrac{s_{1}}{2}(n_{\alpha}-m_{\alpha}(\boldsymbol{\sigma}))-\tfrac{s_{2}}{2}(n_{\alpha}+m_{\alpha}(\boldsymbol{\sigma}))+\textit{i}q_{\alpha}(Xn_{\alpha}+xm_{\alpha}(\boldsymbol{\sigma}))\big]
\nonumber\\
&&\fl\qquad = \prod_{\alpha = 1}^{M}\rme^{-x\partial_{1}\partial_{2}}\Big(\rme^{\rmi x_{1}q_{\alpha}-s_{1}}+\rme^{\rmi x_{2}q_{\alpha}-s_{2}}\Big)^{n_{\alpha}}\;.
\end{eqnarray}

\subsection{Fredholm determinant}
We now return to the generating function $G_{1}$ of (\ref{G1[psi]}), denoting by $G_{1}^{\sharp}$ its approximation under the factorization assumption (\ref{III.23}). The eigenfunctions $\psi_{\boldsymbol{q},\boldsymbol{n}}^{(M)}$ are normalized in such a way that they form an orthonormal basis in the symmetric subspace of $L^{2}(\mathbb{R}^N)$. With this normalization, and using (\ref{III.14}),
\begin{equation}
|\psi_{\boldsymbol{q},\boldsymbol{n}}^{(M)}(0,\ldots,0)|^{2}=N!\det\left(\frac{1}{\frac{1}{2}(n_{j}+n_{k})+i(q_{j}-q_{k})}\right)_{j,k=1,\ldots,M}\;,
\end{equation}
see Eqs. (B.58) and (34) of \cite{D10.2}. From Eq. (B.29) of \cite{D10.2}, the energy $E_{\boldsymbol{q},\boldsymbol{n}}^{(M)}$ of the eigenstate $\psi_{\boldsymbol{q},\boldsymbol{n}}^{(M)}$ is 
\begin{equation}
E_{\boldsymbol{q},\boldsymbol{n}}^{(M)}=\frac{1}{2}\sum_{j=1}^{M}n_{j}q_{j}^{2}-\frac{1}{24}\sum_{j=1}^{M}(n_{j}^{3}-n_{j})\;.
\end{equation}
Finally, the properly normalized sum over the eigenstates is given by 
\begin{equation}
\sum_{r}\equiv\sum_{M=1}^{\infty}\frac{1}{M!}\prod_{j=1}^{M}\Big(\int_{-\infty}^{\infty}\frac{\rmd q_{j}}{2\pi}\sum_{n_{j}=1}^{\infty}\Big)\openone_{\{n=\sum_{j=1}^M n_{j}\}}\;,
\end{equation}
see Eqs. (B.53) and (B.60) of \cite{D10.2}. Combining all, we obtain for $G_{1}^\sharp$ the following expression
\begin{eqnarray}
\label{G1[M,qj,nj]}
&&\fl\qquad G_{1}^\sharp = 1+\sum_{M=1}^{\infty}\frac{1}{M!}\prod_{j=1}^{M}\Big(\int_{-\infty}^{\infty}\frac{\rmd q_{j}}{2\pi}\sum_{n_{j}=1}^{\infty}\Big)\det\left(\frac{1}{\frac{1}{2}(n_{j}+n_{k})+\rmi(q_{j}-q_{k})}\right)_{j,k=1,\ldots,M}\nonumber\\
&&\fl\qquad\qquad\qquad\qquad\qquad\times\prod_{j=1}^{M}\rme^{tn_{j}^{3}/24}\rme^{-x\partial_{1}\partial_{2}}\big[-\rme^{-tq_{j}^{2}/2}\left(\rme^{\rmi x_{1}q_{j}-s_{1}}+\rme^{\rmi x_{2}q_{j}-s_{2}}\right)\big]^{n_{j}}\;.
\end{eqnarray}
We observe that (\ref{G1[M,qj,nj]}) can be rewritten as a Fredholm determinant (see \ref{Appendix Fredholm} for a few basic facts about Fredholm determinants). Indeed, if one introduces the kernel $R$ as
\begin{equation}
\label{kernel R}
R(q,n;q',n')=\frac{1}{2\pi}\,\frac{\rme^{tn^{3}/24}\rme^{-x\partial_{1}\partial_{2}}\big[-\rme^{-\frac{t}{2}q^{2}}\left(\rme^{\rmi x_{1}q-s_{1}}+\rme^{\rmi x_{2}q-s_{2}}\right)\big]^{n}}{\frac{1}{2}(n+n')+i(q-q')}\;,
\end{equation}
then the generating function $G_{1}^\sharp$ is given by
\begin{equation}
G_{1}^\sharp=\det(1+R)\;.
\end{equation}
More explicitly, it holds
\begin{equation}
\fl\qquad G_{1}^\sharp = 1+\sum_{M=1}^{\infty}\frac{1}{M!}\int_{-\infty}^{\infty}\rmd q_{1}\ldots\rmd q_{M}\sum_{n_{1},\ldots,n_{M}=1}^{\infty}\det\big(R(q_{j},n_{j};q_{k},n_{k})\big)_{j,k=1,\ldots,M}\;.
\end{equation}
We perform the summation over the $n_{j}$'s and the integration over the $q_{j}$'s. For this purpose the integrated version for the denominator in (\ref{kernel R}) is used,
\begin{equation}
\frac{1}{\frac{1}{2}(n+n')+\rmi(q-q')}=\int_{0}^{\infty}\rmd z\,\rme^{-z\left[\frac{1}{2}(n+n')+\rmi(q-q')\right]}\;.
\end{equation}
The operator $R$ can then be written as a product of two operators, $R=R_{1}R_{2}$,
\begin{equation}
R(q,n;q',n')=\int_{-\infty}^{\infty}\rmd z\,R_{1}(q,n;z)R_{2}(z;q',n')\;,
\end{equation}
with
\begin{equation}
\fl\qquad R_{1}(q,n;z)=\openone_{\{z>0\}}\,\rme^{-\rmi qz}\rme^{tn^{3}/24}\rme^{-x\partial_{1}\partial_{2}}\big[-\rme^{-\frac{1}{2}z}\rme^{-tq^{2}/2}\left(\rme^{\rmi x_{1}q-s_{1}}+\rme^{\rmi x_{2}q-s_{2}}\right)\big]^{n}\;,
\end{equation}
and
\begin{equation}
R_{2}(z;q',n')=\openone_{\{z>0\}}\,\frac{1}{2\pi}\,\rme^{-\frac{1}{2}n'z}\rme^{\rmi q'z}\;.
\end{equation}
Using $\det(1+R_{1}R_{2})=\det(1+R_{2}R_{1})$, the generating function $G_{1}^\sharp$ becomes equal to a Fredholm determinant of the new operator $N$,
\begin{equation}
G_{1}^\sharp= \det(1+N)\;,
\end{equation}
where $N=R_{2}R_{1}$ with kernel
\begin{equation}
\label{kernel N}
N(z,z')=\int_{-\infty}^{\infty}\rmd q\,\sum_{n=1}^{\infty}R_{2}(z;q,n)R_{1}(q,n;z')\;.
\end{equation}
The variables $q_{j}$ and $n_{j}$, which were previously the variables corresponding to the definition of the Fredholm determinant, are now inside the kernel $N$.

In \ref{Appendix sum q n}, the summation over $n$ and the integration over $q$ in ({\ref{kernel N})
is performed explicitly. Most of the steps are rather similar to the computations done in \cite{D10.1,CLDR10.1,D10.2} in case of the one-point generating function. Note that on face value the sum over $n$ is badly divergent. In terms of the parameter $\alpha=(t/2)^{1/3}$ (equal to $2^{2/3}\lambda$ in the notation of \cite{D10.2}) and of the function $\Phi$ defined in equation (\ref{Phi}), the kernel of $N$ equals
\begin{equation}
\label{Ntilde[N]}
N(z,z') = -\alpha^{-1}\tilde{N}(\alpha^{-1} z,\alpha^{-1} z)\;,
\end{equation}
with
\begin{eqnarray}
\label{kernel Ntilde}
&&\fl\qquad\qquad\tilde{N}(z,z')=\openone_{\{z,z'>0\}}\int_{-\infty}^{\infty}\rmd u\,\rme^{-x\partial_{1}\partial_{2}}\rme^{-(2\alpha)^{-1}(x_{1}\partial_{1}+x_{2}\partial_{2})(\partial_{z}-\partial_{z'})}\nonumber\\
&&\fl\qquad\qquad\qquad\qquad\qquad\qquad\qquad\qquad\Phi(\alpha u-s_{1},\alpha u-s_{2})\Ai(z+u)\Ai(u+z')\;.
\end{eqnarray}
The generating function $G_{1}^\sharp$ is now given by
\begin{equation}
G_{1}^\sharp =\det(1-\tilde{N})\;.
\end{equation}
We will simplify the kernel $\tilde{N}$ and express it in terms of the Airy Hamiltonian $H$ and of the Airy operator $K$.

\subsection{Subtraction of the parabolic profile}
The following transformations are guided to have the shift by $x^{2}/2t$ manifestly visible in $G_1^\sharp$. Using the definition (\ref{X x}) of $X$ and $x$, the expression (\ref{kernel Ntilde}) of the kernel $\tilde{N}$ rewrites as
\begin{eqnarray}
\label{Ntilde[X]}
&&\fl\qquad\tilde{N}(z,z')=\openone_{\{z,z'>0\}}\,\rme^{-x\partial_1\partial_2}\rme^{\frac{x}{2\alpha}(\partial_{1}-\partial_{2})(\partial_{z}-\partial_{z'})}\nonumber\\
&&\fl\qquad\qquad\qquad\int_{-\infty}^{\infty}\rmd u\,\rme^{-\frac{X}{2\alpha}(\partial_{1}+\partial_{2})(\partial_{z}-\partial_{z'})}\Phi(\alpha u-s_{1},\alpha u-s_{2})\Ai(z+u)\Ai(u+z')\;.
\end{eqnarray}
We note that   $\partial_{1}+\partial_{2}$ in (\ref{Ntilde[X]}) can replaced by $-\alpha^{-1}\partial_{u}$, where the derivative $\partial_{u}$ acts only on $\Phi(\alpha u-s_{1},\alpha u-s_{2})$ and not on the product of Airy functions. One can then make $\partial_{u}$ to act only on the product of Airy functions by integrating by parts. Thereby
\begin{eqnarray}
&&\fl\qquad\tilde{N}(z,z')=\openone_{\{z,z'>0\}}\,\rme^{-x\partial_{1}\partial_{2}}\rme^{\frac{x}{2\alpha}(\partial_{1}-\partial_{2})(\partial_{z}-\partial_{z'})}\nonumber\\
&&\fl\qquad\qquad\int_{-\infty}^{\infty}\rmd u\,\Phi(\alpha u-s_{1},\alpha u-s_{2})\rme^{-(2\alpha^{2})^{-1}X\partial_{u}(\partial_{z}-\partial_{z'})}\Ai(z+u)\Ai(u+z')\;.
\end{eqnarray}
Since the Airy function is a solution of the differential equation $\Ai''(u)=u\Ai(u)$, we have
\begin{equation}
\partial_{u}(\partial_{z}-\partial_{z'})\Ai(z+u)\Ai(u+z')=(z-z')\Ai(z+u)\Ai(u+z')
\end{equation}
and hence
\begin{equation}
\rme^{a[\partial_{u}(\partial_{z}-\partial_{z'})-(z-z')]}\Ai(z+u)\Ai(u+z')=\Ai(z+u)\Ai(u+z')\;.
\end{equation}
The commutator of $\partial_{u}(\partial_{z}-\partial_{z'})$ and $(z-z')$ is equal to $2\partial_{u}$ and it commutes with both $\partial_{u}(\partial_{z}-\partial_{z'})$ and $(z-z')$. Thus the Baker-Campbell-Hausdorff formula terminates and
\begin{equation}
\fl\qquad\qquad \rme^{-a\partial_{u}(\partial_{z}-\partial_{z'})}\Ai(z+u)\Ai(u+z')=\rme^{-a(z-z')}\rme^{a^{2}\partial_{u}}\Ai(z+u)\Ai(u+z')\;.
\end{equation}
We use this property in the expression of the kernel $\tilde{N}$, and integrate by parts to make $\partial_{u}$ act again on $\Phi(\alpha u-s_{1},\alpha u-s_{2})$ with the result
\begin{eqnarray}
\fl\qquad\qquad\tilde{N}(z,z') &=& \openone_{\{z,z'>0\}}\,\rme^{-x\partial_{1}\partial_{2}}\rme^{\frac{x}{2\alpha}(\partial_{1}-\partial_{2})(\partial_{z}-\partial_{z'})}\int_{-\infty}^{\infty}\rmd u\,\rme^{-(2\alpha^{2})^{-1}X(z-z')}\Ai(z+u)\nonumber\\
\fl\qquad\qquad && \times\Ai(u+z')\rme^{-(4\alpha^{4})^{-1}X^{2}\partial_{u}}\Phi(\alpha u-s_{1},\alpha u-s_{2})\;.
\end{eqnarray}
Using the commutation relation
\begin{equation}
\rme^{a(\partial_{z}-\partial_{z'})}\rme^{-b(z-z')}=\rme^{-2ab}\rme^{-b(z-z')}\rme^{a(\partial_{z}-\partial_{z'})}\;,
\end{equation}
one obtains
\begin{eqnarray}
\fl\qquad \tilde{N}(z,z') &=& \openone_{\{z,z'>0\}}\,\rme^{-(2\alpha^{2})^{-1}X(z-z')}\rme^{-x\partial_{1}\partial_{2}}\rme^{\frac{x}{2\alpha}(\partial_{1}-\partial_{2})(\partial_{z}-\partial_{z'})} \int_{-\infty}^{\infty}\rmd u\,\Ai(z+u)\nonumber\\
\fl\qquad && \times\Ai(u+z')\rme^{-(2\alpha^{3})^{-3}Xx(\partial_{1}-\partial_{2})}\rme^{-(4\alpha^{4})^{-1}X^{2}\partial_{u}}\Phi(\alpha u-s_{1},\alpha u-s_{2})\;.
\end{eqnarray}
Since $\rme^{a\partial}$ acts as a shift operator, we arrive at
\begin{eqnarray}
\fl \tilde{N}(z,z') &=& \openone_{\{z,z'>0\}}\,\rme^{-(2\alpha^{2})^{-1}X(z-z')}\rme^{-x\partial_{1}\partial_{2}}\rme^{\frac{x}{2\alpha}(\partial_{1}-\partial_{2})(\partial_{z}-\partial_{z'})}\rme^{-(4\alpha^{3})^{-1}x_{2}^{2}(\partial_{1}+\partial_{2})}\nonumber\\
\fl && \times\int_{-\infty}^{\infty}\rmd u\,\Ai(z+u)\Ai(u+z')\,\Phi\left(\alpha u-s_{1}- \tfrac{1}{2t}x_{1}^{2},\alpha u-s_{2}-\tfrac{1}{2t}x_{2}^{2}\right)\;.
\end{eqnarray}
The factor $\exp\big(-(2\alpha^{2})^{-1}X(z-z')\big)$ can be eliminated by a similarity transformation of the kernel, it does not contribute to the Fredholm determinant. We define the shifted generating function
\begin{equation}
G_2^\sharp(u,v;x,t) = G_1^\sharp(u-\tfrac{1}{2t}x_{1}^{2},v-\tfrac{1}{2t}x_{2}^{2};x_1,x_2,t)\;.
\end{equation}
and find that $G_{2}^\sharp$ can be written as
\begin{equation}
\label{G1[L]}
G_{2}^\sharp=\det(1-L)\;,
\end{equation}
where the kernel $L$ is given by
\begin{eqnarray}
\label{kernel L}
\fl\qquad L(z,z') &=& \openone_{\{z,z'>0\}}\,\rme^{-x\partial_{1}\partial_{2}}\rme^{\frac{x}{2\alpha}(\partial_{1}-\partial_{2})(\partial_{z}-\partial_{z'})}\rme^{-(4\alpha^{3})^{-1}x^{2}(\partial_{1}+\partial_{2})}\nonumber\\
\fl\qquad && \times\int_{-\infty}^{\infty}\rmd u\,\Ai(z+u)\Ai(u+z')\,\Phi\left(\alpha u-s_{1},\alpha u-s_{2}\right)\;.
\end{eqnarray}
In this form, the kernel depends only on $x$ (and $t$), as to be expected from the discussion in the introduction.

\subsection{Rewriting of the kernel \texorpdfstring{$L$}{L} in terms of the Airy Hamiltonian}
The next step is to eliminate the operator $\partial_{1}\partial_{2}$ from the expression (\ref{kernel L}) for the kernel $L$, for which purpose one writes
\begin{eqnarray}
\fl\qquad L(z,z') &=& \openone_{\{z,z'>0\}}\int_{-\infty}^{\infty}\rmd u\,\rme^{\frac{x}{2}(\partial_{1}^{2}+\partial_{2}^{2})}\rme^{-\frac{x}{2}(\partial_{1}+\partial_{2})^{2}}\rme^{-(4\alpha^{3})^{-1}x^{2}(\partial_{1}+\partial_{2})}\rme^{\frac{x}{2\alpha}(\partial_{1}-\partial_{2})(\partial_{z}-\partial_{z'})}\nonumber\\
\fl\qquad && \times\Phi (\alpha u-s_{1},\alpha u-s_{2})\Ai(z+u)\Ai(u+z')\;.
\end{eqnarray}
In this expression, we can replace $(\partial_{1}+\partial_{2})^{2}$ by $\alpha^{-2}\partial_{u}^{2}$, where $\partial_{u}$ acts only on $\Phi (\alpha u-s_{1},\alpha u-s_{2})$ and not on the product of Airy functions. Then, integrating by parts, $\partial_{u}$ acts on the product of Airy functions instead. Thereby
\begin{eqnarray}
\fl\qquad L(z,z') &=& \openone_{\{z,z'>0\}}\int_{-\infty}^{\infty}\rmd u\,\rme^{\frac{x}{2}(\partial_{1}^{2}+\partial_{2}^{2})}\rme^{-(4\alpha^{3})^{-1}x^{2}(\partial_{1}+\partial_{2})}\rme^{\frac{x}{2\alpha}(\partial_{1}-\partial_{2})(\partial_{z}-\partial_{z'})}\nonumber\\
\fl\qquad && \times\Phi (\alpha u-s_{1},\alpha u-s_{2})\rme^{- (2\alpha^{2})^{-1}x \partial_{u}^{2}}\Ai(z+u)\Ai(u+z')\;.
\end{eqnarray}
But $\partial_{u}$ acting on $\Ai(z+u)\Ai(u+z')$ is the same as $\partial_{z}+\partial_{z'}$. Thus, one can replace $\exp\!\big(- (2\alpha^{2})^{-1}x \partial_{u}^{2}\big)$ by $\exp\!\big(- (2\alpha^{2})^{-1}x \partial_{u}(\partial_{z}+\partial_{z'})\big)$. After integrating by parts again, we replace $\exp\!\big( (2\alpha^{2})^{-1}x \partial_{u}(\partial_{z}+\partial_{z'})\big)$ by $\exp\!\big(-\frac{x}{2\alpha}(\partial_{1}+\partial_{2})(\partial_{z}+\partial_{z'})\big)$ and obtain
\begin{eqnarray}
\fl\qquad L(z,z') &=& \openone_{\{z,z'>0\}}\int_{-\infty}^{\infty}\rmd u\,\rme^{\frac{x}{2}(\partial_{1}^{2}+\partial_{2}^{2})}\rme^{-(4\alpha^{3})^{-1}x^{2}(\partial_{1}+\partial_{2})}\rme^{-\frac{x}{\alpha}(\partial_{1}\partial_{z'}+\partial_{2}\partial_{z})}\nonumber\\
\fl\qquad && \times\Phi (\alpha u-s_{1},\alpha u-s_{2})\Ai(z+u)\Ai(u+z')\;.
\end{eqnarray}
We now introduce a dummy integration to achieve that $\alpha u-s_{1}$, $\alpha u-s_{2}$, $\Ai(z+u)$ and $\Ai(u+z')$ do not all depend on the same variable $u$. It holds
\begin{eqnarray}
\fl\qquad L(z,z') &=& \openone_{\{z,z'>0\}}\int_{-\infty}^{\infty}\rmd u\,\rmd v\,\delta(u-v)\rme^{\frac{x}{2}(\partial_{1}^{2}-\frac{2}{\alpha}\partial_{1}\partial_{z'})}\rme^{\frac{x}{2}(\partial_{2}^{2}-\frac{2}{\alpha}\partial_{2}\partial_{z})}\rme^{-(4\alpha^{3})^{-1}x^{2}(\partial_{1}+\partial_{2})}\nonumber\\
\fl\qquad && \times\Phi (\alpha u-s_{1},\alpha v-s_{2})\Ai(z+v)\Ai(u+z')\;.
\end{eqnarray}
Note that
\begin{equation}
\fl\left[\left(\partial_{1}^{2}-2\alpha^{-1}\partial_{1}\partial_{z'}\right)+ \alpha^{-2}(z'+H_{u})\right]\Phi (\alpha u-s_{1},\alpha v-s_{2})\Ai(z+v)\Ai(u+z')=0\;,
\end{equation}
where $H_{u}=-\partial_{u}^{2} +u$ denotes the Airy Hamiltonian acting on the variable $u$. The commutation relation
\begin{equation}
\big[(\partial_{1}^{2} - 2\alpha^{-1}\partial_{1}\partial_{z'}),\alpha^{-2}(z'+H_{u})\big]=-2\alpha^{-3}\partial_{1}
\end{equation}
together with the Baker-Campbell-Hausdorff formula implies
\begin{eqnarray}
\fl\qquad &&\rme^{\frac{x}{2}(\partial_{1}^{2}-\frac{2}{\alpha}\partial_{1}\partial_{z'})}\,\Phi (\alpha u-s_{1},\alpha v-s_{2})\Ai(z+v)\Ai(u+z')\nonumber\\
\fl\qquad && = \rme^{- (2\alpha^{2})^{-1}x (z'+H_{u})}\rme^{(4\alpha^{3})^{-1}x^{2}\partial_{1}}\,\Phi (\alpha u-s_{1},\alpha v-s_{2})\Ai(z+v)\Ai(u+z')\;.
\end{eqnarray}
A corresponding relation is obtained by interchanging the roles of $z$, $s_{1}$, $u$ and of $z'$, $s_{2}$, $v$. Using both the kernel $L$ becomes
\begin{eqnarray}
\fl\qquad L(z,z') &=& \openone_{\{z,z'>0\}}\int_{-\infty}^{\infty}\rmd u\,\rmd v\,\delta(u-v)\rme^{- (2\alpha^{2})^{-1}x (z+z')}\rme^{- (2\alpha^{2})^{-1}x (H_{u}+H_{v})}\nonumber\\
\fl\qquad && \times\Phi(\alpha u-s_{1},\alpha v-s_{2})\Ai(z+v)\Ai(u+z')\;.
\end{eqnarray}
A final integration by parts over $u$ and $v$ yields
\begin{eqnarray}
\label{kernel L final}
\fl\qquad L(z,z') &=& \openone_{\{z,z'>0\}}\int_{-\infty}^{\infty}\rmd u\,\rmd v\,\langle u|\rme^{-\alpha^{-2}xH}|v\rangle\nonumber\\
\fl\qquad && \times\Phi(\alpha u-s_{1},\alpha v-s_{2})\rme^{- (2\alpha^{2})^{-1}x (z+z')}\Ai(z+v)\Ai(u+z')\;.
\end{eqnarray}
This last expression is an operator product of the form ABA. Hence using the cyclicity of the determinant one arrives at
\begin{equation}
G_{2}^\sharp = \det(1-L)=\det(1-\tilde{L})\;,
\end{equation}
where
\begin{equation}
\label{kernel Ltilde}
\langle u |\tilde{L} | v\rangle= \langle u|\rme^{-\alpha^{-2}xH}|v\rangle\Phi (\alpha u-s_{1},\alpha v-s_{2}) \langle u |\rme^{\alpha^{-2}xH}K |v\rangle\;.
\end{equation}
We conclude that $\tilde{L} = Q\,\rme^{\alpha^{-2}|x|H}K$, as to be shown.


\section{Finite time probability density function}
\label{Section finite time}
\setcounter{equation}{0}
The two-point free energy fluctuations can be written as
\begin{equation}
F(x_{j},t)=\frac{1}{24}t+\frac{1}{2t}x_{j}^{2}+\xi_{j}(t),\quad j=1,2\;,
\end{equation}
with random amplitudes $\xi_{j}(t)$. In the factorization approximation, we know already that $\xi_{j}(t)=\O{t^{1/3}}$. The joint distribution depends only on $|x_{1}-x_{2}|$, and nondegenerate correlations occur for a separation of order $t^{2/3}$. Following the procedure in \cite{CLDR10.1}, we would like to extract the underlying pdf from the generating function $G^{\sharp}$. Let us denote by $\rho^{\sharp}_{t}(w_{1},w_{2})$ the approximate joint pdf of $\xi_{1}(t)$ and $\xi_{2}(t)$, and write it as the convolution with two independent Gumbel densities $F_{\text{Gu}}'$:
\begin{equation}
\label{rhot[gt]}
\rho_{t}^\sharp(w_{1},w_{2})=\int_{-\infty}^{\infty}\rmd v_{1}\,\rmd v_{2}\,F_{\text{Gu}}'(w_{1}-v_{1})F_{\text{Gu}}'(w_{2}-v_{2})g_{t}(v_{1},v_{2})\;,
\end{equation}
where
\begin{equation}
F_{\text{Gu}}(w)=\exp(-\rme^{-w})\;,
\end{equation}
compare with the one-point distribution (\ref{distrib 1pt F(x,t)}). Eq. (\ref{rhot[gt]}) defines the yet to be determined function $g_{t}$, which is normalized to $1$ by construction. From numerical solutions in the one-point case, we know that $g_{t}$ of (\ref{gt 1pt}) is in general not everywhere positive. This implies that $g_{t}$ of (\ref{rhot[gt]}) will not be a pdf, in general.

The generating function $G_{2}^\sharp$ reads
\begin{equation}
G_{2}^\sharp=\int_{-\infty}^{\infty}\rmd w_{1}\,\rmd w_{2}\,\rho_{t}(w_{1},w_{2})\exp(-\rme^{-s_{1}-w_{1}}-\rme^{-s_{2}-w_{2}})\;.
\end{equation}
Inserting (\ref{rhot[gt]}) and performing the integration over $w_{1}$ and $w_{2}$ yields
\begin{equation}
\label{eq gt L}
G_{2}^\sharp=\int_{-\infty}^{\infty}\rmd v_{1}\,\rmd v_{2}\,g_{t}(v_{1},v_{2})\,\frac{\rme^{v_{1}}}{\rme^{v_{1}}+\rme^{-s_{1}}}\,\frac{\rme^{v_{2}}}{\rme^{v_{2}}+\rme^{-s_{2}}}\;.
\end{equation}
We analytically continue on both sides from $\rme^{-s_{j}}$ to $-\rme^{a_{j}}-\rmi\sigma_{j}\varepsilon$, $\varepsilon>0$, $j=1,2$. From the identity
\begin{equation}
\lim_{\varepsilon\to0}\sum_{\sigma=\pm1}\,\frac{\sigma \rme^{v}}{\rme^{v}-\rme^{a}-\rmi\sigma\varepsilon}=2\mathit{i}\pi\delta(v-a)\;,
\end{equation}
the left side of (\ref{eq gt L}) multiplied by $\sigma_{1}\sigma_{2}$ and summed over $\sigma_{1},\sigma_{2}=\pm1$ yields in the limit $\varepsilon\to0$
\begin{equation}
-4\pi^{2}g_{t}(a_{1},a_{2})\;.
\end{equation}
On the right side of (\ref{eq gt L}), using $G_{2}^\sharp=\det(1-L)$ with $L$ given by (\ref{kernel L final}), we obtain the sum of four Fredholm determinants with operators $L_{\sigma_{1},\sigma_{2}}$, $\sigma_{j}=\pm1$. To compute these kernels, we have to take the limit $\varepsilon\to0$ of
\begin{equation}
\int_{-\infty}^{\infty}\rmd u_{1}\,\rmd u_{2}\,h(u_{1},u_{2})\frac{-\rme^{\alpha u_{1}+a_{1}}-\rme^{\alpha u_{2}+a_{2}}-\rmi\varepsilon(\sigma_{1}\rme^{\alpha u_{1}}+\sigma_{2}\rme^{\alpha u_{2}})}{1-\rme^{\alpha u_{1}+a_{1}}-\rme^{\alpha u_{2}+a_{2}}-\rmi\varepsilon(\sigma_{1}\rme^{\alpha u_{1}}+\sigma_{2}\rme^{\alpha u_{2}})}\;,
\end{equation}
with $h$ general at this stage. Using
\begin{equation}
\frac{1}{y-\rmi\sigma\varepsilon}=\mathcal{P}\Big(\frac{1}{y}\Big)+\rmi\sigma\pi\delta(y)\;,
\end{equation}
one obtains
\begin{eqnarray}
\fl\qquad L_{\sigma_{1},\sigma_{2}}(z,z')  &=& \openone_{\{z,z'>0\}}\int_{-\infty}^{\infty}\rmd u\,\rmd v\,\langle u|\rme^{-\alpha^{-2}xH}|v\rangle \rme^{-(2\alpha^{2})^{-1}x(z+z')}\Ai(z+u)\Ai(v+z')\nonumber\\
\fl\qquad && \times\Big(\mathcal{P}\Big(\frac{-\rme^{\alpha u_{1}+a_{1}}-\rme^{\alpha u_{2}+a_{2}}}{1-\rme^{\alpha u_{1}+a_{1}}-\rme^{\alpha u_{2}+a_{2}}}\Big)+i\pi\sgn(\sigma_{1}\rme^{\alpha u_{1}}+\sigma_{2}\rme^{\alpha u_{2}})\nonumber\\
\fl\qquad && \times(-\rme^{\alpha u_{1}+a_{1}}-\rme^{\alpha u_{2}+a_{2}})\delta(1-\rme^{\alpha u_{1}+a_{1}}-\rme^{\alpha u_{2}+a_{2}})\Big)\;,
\end{eqnarray}
and
\begin{equation}
g_{t}(a_{1},a_{2})=-(2\pi)^{-2}\sum_{\sigma_{1},\sigma_{2}=\pm1}\sigma_{1}\sigma_{2}\det(1-L_{\sigma_{1},\sigma_{2}})\;.
\end{equation}

In the one-point case, inserting the corresponding analytic continuation in (\ref{kernel M}) yields a one-dimensional projection. This further simplifies the expression for the pdf, compare with (\ref{gt 1pt}). For the two-point case, no further simplification seems to be available.


\begin{section}{Conclusions}
Recently the probability distribution of the free energy of the point-to-point continuum directed polymer has been computed exactly, using the approximation through the weakly asymmetric simple exclusion process \cite{ACQ10,SS10.3}. This result could be reproduced by using replicas and the complete eigenfunction expansion of the propagator of the attractive $\delta$-Bose gas on the line \cite{CLDR10.1,D10.1}.

In our contribution we studied the joint pdf of $F(x_1,t), F(x_2,t)$. In this case the approximation through the weakly asymmetric exclusion process, while still valid, is no longer computable and we have to rely on the replica method, which yields a particular generating function. Invoking a specific factorization, the result is expressed as a Fredholm determinant. In fact, the corresponding operator has a structure rather similar to the case of one point, compare (\ref{G Fredholm 1pt}) and (\ref{G Fredholm 2pt}).

Equipped with this information, we established the large time limit of the pdf yielding a result in agreement with lattice directed polymers at zero temperature. In 1+1 dimensions all models with a short range disorder potential are expected to flow to the zero temperature fixed point. Our result further supports this claim.

The free energy of the point-to-point continuum directed polymer is isomorphic to the solution of the KPZ equation with sharp wedge initial data. Thus we have automatically determined the joint pdf of the heights $h(x_1,t)$, $h(x_2,t)$ of the KPZ equation for large times, in particular the height-height correlation function. This function has been measured recently for droplet growth in a thin film of turbulent liquid crystal \cite{TS10.1}. The experimental curve agrees very well with the theoretical prediction in the limit $t \to \infty$.

\subsection*{Acknowledgements}
It is a pleasure to thank Pierre Le Doussal, Michael Pr\"{a}hofer, and Tomohiro Sasamoto for constructive discussions.
\end{section}


\begin{appendix}
\section{Fredholm determinants}
\label{Appendix Fredholm}
\setcounter{equation}{0}
Let us first consider the case of a finite $n\times n$ matrix $A$. The Taylor expansion of the determinant of $1+zA$ is given by the von Koch formula \cite{B10.1} in terms of the minors of $A$. It holds
\begin{equation}
\label{von Koch}
\det(1+zA)=\sum_{m=0}^{n}\frac{z^{m}}{m!}\sum_{i_{1},\ldots,i_{m}=1}^{n}\det\big(A_{i_{j},i_{k}}\big)_{1\leq j,k\leq m}\;.
\end{equation}
If the matrix $A$ is now replaced by an integral operator $A$ with kernel $A(u,v)=\langle u|A|v\rangle$, von Koch formula (\ref{von Koch}) formally rewrites as
\begin{equation}
\label{def Fredholm}
\det(1+zA)\equiv\sum_{m=0}^{\infty}\frac{z^{m}}{m!}\int \rmd u_{1}\,\ldots\,\rmd u_{m}\,\det\big(A(u_{j},u_{k})\big)_{1\leq j,k\leq m}\;.
\end{equation}
Of course, $\rmd u$ could mean a more general summation procedure. In particular, it could refer to summation over some discrete index and integration over $\mathbb{R}$. (\ref{def Fredholm}) can be used as the \textit{definition} of the Fredholm determinant $\det(1+zA)$. In order for this definition to make sense, the operator $A$ is required to be \textit{trace-class}, we refer to \cite{S05.2} for details. If so, the logarithm of the Fredholm determinant is given by
\begin{eqnarray}
\label{log Fredholm}
\fl\qquad && \log\det(1+zA) = \tr\log(1+zA)\nonumber\\
\fl\qquad && = \sum_{m=1}^{\infty}\frac{(-1)^{m-1}z^{m}}{m}\int \rmd u_{1}\,\ldots\,\rmd u_{m}\,A(u_{1},u_{2})A(u_{2},u_{3})\ldots A(u_{m},u_{1})\;.
\end{eqnarray}
Another useful identity for Fredholm determinants is the cycle property. If $A$ and $B$ are Hilbert-Schmidt operators (\textit{i.e.} $\tr AA^{*}<\infty$ and $\tr BB^{*}<\infty$), then
\begin{equation}
\label{det(1+AB)=det(1+BA)}
\det(1+AB)=\det(1+BA)\;.
\end{equation}
This property allows to exchange the roles of integrations which are inside the defining kernel of $AB$ with the integration corresponding to the Fredholm determinant. We emphasize that the two kernels $AB$ and $BA$ do not necessarily act on the same space.

Numerical evaluations of Fredholm determinants can be performed by discretizing the integrals in (\ref{def Fredholm}). Using the von Koch formula (\ref{von Koch}), the evaluation of a Fredholm determinant is thereby reduced to the computation of the determinant of a finite matrix. We refer to \cite{B10.1,B10.2} for an illuminating discussion and precise error estimates.


\section{Airy operator and Airy Hamiltonian}
\label{Appendix Airy}
\setcounter{equation}{0}
We recall the definition of the Airy Hamiltonian $H$ and of the Airy operator $K$. We work in the space of complex-valued square integrable functions $L^{2}(\mathbb{R})$ with the scalar product
\begin{equation}
\langle f|g\rangle=\int_{-\infty}^{\infty}\rmd u\,f(u)^*g(u)\;.
\end{equation}
The Airy Hamiltonian $H$ is defined by
\begin{equation}
H=-(\partial_{u})^{2}+u\;.
\end{equation}
If necessary, we write $H$ as $H_{u}$ to indicate the variable on which the Airy operator is acting. The Airy function is the solution of the differential equation
\begin{equation}
\Ai''(u)=u\Ai(u)\;
\end{equation}
such that $\Ai(u)\to0$ as $u\to\infty$. Setting
\begin{equation}
\phi_{z}(u)=\Ai(u-z)\;,
\end{equation}
one notes that $\phi_{z}$ satisfies the eigenvalue equation
\begin{equation}
H\phi_{z}=z\phi_{z}\;.
\end{equation}
In addition,
\begin{equation}
\int_{-\infty}^{\infty}\rmd z\,\Ai(u-z)\Ai(u'-z)=\delta(u-u')\;.
\end{equation}
In Dirac notation this completeness relation reads
\begin{equation}
\openone=\int_{-\infty}^{\infty}\rmd z\,|\phi_{z}\rangle\langle\phi_{z}|\;.
\end{equation}
Hence the Airy Hamiltonian has the spectral representation
\begin{equation}
H=\int_{-\infty}^{\infty}\rmd z\,z\,|\phi_{z}\rangle\langle\phi_{z}|\;.
\end{equation}
The projection to all negative energy states defines the Airy operator
\begin{equation}
K=\int_{-\infty}^{0}\rmd z\,|\phi_{z}\rangle\langle\phi_{z}|\;.
\end{equation}
In particular, one has $K=K^{*}$, $K^{2}=K$, and obviously $[K,H] =0$. In position representation, the Airy kernel writes
\begin{equation}
\fl\qquad \langle u|K|v\rangle=\int_{0}^{\infty}\rmd z\Ai(u+z)\Ai(z+v)=\frac{\Ai(u)\Ai'(v)-\Ai'(u)\Ai(v)}{u-v}\;.
\end{equation}
We also define the projection onto the spatial interval $[a,\infty)$ by
\begin{equation}
P_{a}=\int_{a}^{\infty}\rmd u\,|u\rangle\langle u|\;.
\end{equation}
The operator $P_{a}KP_{a}$ is trace class for all $a>-\infty$. $F_{2}(a) = \det(1 - P_{a}KP_{a})$ is by definition the Tracy-Widom distribution \cite{TW94.1} corresponding to the Gaussian Unitary Ensemble of random matrices.


\section{Integration over \texorpdfstring{$q$}{q} and summation over \texorpdfstring{$n$}{n} in the kernel \texorpdfstring{$N$}{N}}
\label{Appendix sum q n}
\setcounter{equation}{0}
We start from the explicit expression
\begin{eqnarray}
N(z,z') &=& \openone_{\{z,z'>0\}}\int_{-\infty}^{\infty}\frac{\rmd q}{2\pi}\sum_{n=1}^{\infty}\rme^{\rmi q(z-z')}\rme^{t n^{3}/24}\rme^{-x\partial_{1}\partial_{2}}\nonumber\\
&& \times\big[-\rme^{-\frac{1}{2}(z+z')}\rme^{-t q^{2}/2}\left(\rme^{\rmi x_{1}q-s_{1}}+\rme^{\rmi x_{2}q-s_{2}}\right)\big]^{n}\;.
\end{eqnarray}
In order to perform the summation over $n$ and the integration over $q$, we insert the classical relation
\begin{equation}\label{D.2}
\rme^{t n^{3}/24}=\int_{-\infty}^{\infty}\rmd u\,\Ai(u)\rme^{(t/8)^{1/3}nu}\;.
\end{equation}
Besides the Airy function there are infinitely many other functions which satisfy (\ref{D.2}). One concrete example would be
\begin{equation}
\mathrm{Ai}(u) + \sin(\pi u)\rme^{-u^{2}/2}\;.
\end{equation}
Our choice is determined by being the only one which provides the correct one-point result. Using the binomial theorem to expand the term of power $n$, one obtains
\begin{eqnarray}
\fl\qquad N(z,z') &=& \openone_{\{z,z'>0\}}\int_{-\infty}^{\infty}\frac{\rmd q\,\rmd u}{2\pi}\sum_{n=1}^{\infty}\sum_{k=0}^{n}\Ai(u)\rme^{\rmi q[z-z'+x_{1}k+x_{2}(n-k)]}\nonumber\\
\fl\qquad && \times\C{n}{k}\left(-\rme^{(t/8)^{1/3}u}\rme^{-\frac{1}{2}(z+z')}\rme^{-t q^{2}/2}\right)^{n}\rme^{-x\partial_{1}\partial_{2}}\rme^{-s_{1}k-s_{2}(n-k)}\;.
\end{eqnarray}
We introduce the parameter $\alpha=(t/2)^{1/3}$ and perform the change of variable $u\to u+2^{2/3}\alpha^{2}q^{2}+2^{-1/3}\alpha^{-1}(z+z')$ in the integral. This results in
\begin{eqnarray}
\fl\qquad N(z,z') &=& \openone_{\{z,z'>0\}}\int_{-\infty}^{\infty}\frac{\rmd q\,\rmd u}{2\pi}\sum_{n=1}^{\infty}\sum_{k=0}^{n}\Ai\left(u+2^{2/3}\alpha^{2}q^{2}+2^{-1/3}\alpha^{-1}(z+z')\right)\nonumber\\
\fl\qquad && \times \rme^{\rmi q[z-z'+x_{1}k+x_{2}(n-k)]}\C{n}{k}(-1)^{n}\rme^{2^{-2/3}\alpha un}\rme^{-x\partial_{1}\partial_{2}}\rme^{-s_{1}k-s_{2}(n-k)}\;.
\end{eqnarray}
Using the relation 
\begin{eqnarray}
&& \int_{-\infty}^{\infty}\frac{\rmd q}{2\pi}\Ai(aq^{2}+b)\rme^{\rmi cq}\nonumber\\
&& \qquad =2^{-1/3}a^{-1/2}\Ai\left(2^{-2/3}\left(b+a^{-1/2}c\right)\right)\Ai\left(2^{-2/3}\left(b-a^{-1/2}c\right)\right)\;,
\end{eqnarray}
see \cite{VSdI97.1}, the integration over $q$ can be performed. One obtains
\begin{eqnarray}
\fl\qquad N(z,z') &=& \openone_{\{z,z'>0\}}\int_{-\infty}^{\infty}\rmd u\sum_{n=1}^{\infty}\sum_{k=0}^{n}\C{n}{k}(-1)^{n}\rme^{2^{-2/3}\alpha un}\rme^{-x\partial_{1}\partial_{2}}\rme^{-s_{1}k-s_{2}(n-k)}
\nonumber\\
\fl\qquad && \times2^{-2/3}\alpha^{-1}\Ai\left(2^{-2/3}u+\alpha^{-1}z+(2\alpha)^{-1}[x_{1}k+x_{2}(n-k)]\right)
\nonumber\\
\fl\qquad && \times\Ai\left(2^{-2/3}u+\alpha^{-1}z'-(2\alpha)^{-1}[x_{1}k+x_{2}(n-k)]\right)\;.
\end{eqnarray}
We change variables as $u\to2^{2/3}u$, $z\to\alpha z$, and use the relation
\begin{equation}
f(z+a)=\exp[a\partial_{z}]f(z)\;
\end{equation}
to move $n$ and $k$ out of the Airy functions. In preparation for the summation over $n$ and $k$, one finds
\begin{eqnarray}
\fl\qquad \alpha N(\alpha z,\alpha z') &=& \openone_{\{z,z'>0\}}\int_{-\infty}^{\infty}\rmd u\sum_{n=1}^{\infty}\sum_{k=0}^{n}\C{n}{k}(-1)^{n}\rme^{\alpha un}\rme^{-x\partial_{1}\partial_{2}}\rme^{-s_{1}k-s_{2}(n-k)}\nonumber\\
\fl\qquad && \times \rme^{(2\alpha)^{-1}(x_{1}k+x_{2}(n-k))(\partial_{z}-\partial_{z'})}\Ai(u+z)\Ai(u+z')\;.
\end{eqnarray}
Noticing the identity
\begin{equation}
\fl\quad \rme^{(2\alpha)^{-1}(x_{1}k+x_{2}(n-k))(\partial_{z}-\partial_{z'})}\rme^{-s_{1}k-s_{2}(n-k)}=\rme^{-(2\alpha)^{-1}(x_{2}\partial_{2}+x_{1}\partial_{1})(\partial_{z}-\partial_{z'})}\rme^{-s_{1}k-s_{2}(n-k)}\;,
\end{equation}
the summations over $k$ and $n$ finally yield the expression (\ref{Ntilde[N]}), (\ref{kernel Ntilde}) for the kernel $N$ 
with $\Phi$ is defined as in Eq. (\ref{Phi}).

\end{appendix}


\end{document}